\documentclass[a4paper, aps, pra, reprint, nofootinbib]{revtex4-1}
\usepackage[cm]{fullpage}
\usepackage[T1]{fontenc}
\usepackage{amssymb, amsmath, amsthm}
\makeatletter
\def\amsbb{\use@mathgroup \M@U \symAMSb}
\makeatother
\usepackage[mathscr]{euscript}
\usepackage{mathtools}
\usepackage{verbatim}
\usepackage{bbold}
\usepackage{subfig}
\usepackage{xcolor}
\definecolor{darkred}{RGB}{150, 0, 0}
\definecolor{darkgreen}{RGB}{0, 130, 0}
\definecolor{darkblue}{RGB}{0, 0, 200}
\usepackage{hyperref}
\hypersetup{colorlinks, breaklinks, linkcolor=darkred, urlcolor=darkblue, citecolor=darkgreen}
\newcommand{\prlparagraph}[1]{\section{#1}}

\newcommand{\nbox}[2][9]{\hspace{#1pt} \mbox{#2} \hspace{#1pt}}
\newtheorem{lem}{Lemma}[section]

\newtheorem{thm}{Theorem}
\newtheorem{cor}{Corollary}[section]

\DeclareMathOperator{\tr}{tr}

\DeclareMathOperator{\id}{id}

\newcommand{\I}{\mathbb{1}}
\newcommand{\X}{\mathsf{X}}
\newcommand{\Y}{\mathsf{Y}}
\newcommand{\Z}{\mathsf{Z}}
\def \diracspacing {0.7pt}
\newcommand{\ket}[1]{| \hspace{\diracspacing} #1 \rangle} 
\newcommand{\ketbra}[2]{| \hspace{\diracspacing} #1 \rangle \langle #2 \hspace{\diracspacing} |} 
\newcommand{\ketbraq}[1]{\ketbra{#1}{#1}} 
\newcommand{\tran}[0]{^\textnormal{\tiny{T}}}
\newcommand{\hc}{^{\dagger}}
\newcommand{\norm}[2][]{#1| \! #1| #2 #1| \! #1|}
\newcommand{\ave}[2][]{#1\langle #2 #1\rangle}
\newcommand{\abs}[2][]{#1| #2 #1|}
\newcommand{\cL}{\mathcal{L}}

\newcommand{\sH}{\mathscr{H}}

\usepackage{tikz}
\tikzset{>=stealth}
\newcommand{\qval}{4 + \alpha^{2}}
\newcommand{\LambdaB}[1]
{
\begin{equation*}
\Lambda_{B}(\rho) := #1 \ave{ \I, \rho } \, \I + #1 \ave{ E_{\X}, \rho} \, \X + #1 \ave{ E_{\Y}, \rho} \, \Y + #1 \ave{ E_{\Z}, \rho } \, \Z,
\end{equation*}
}
\newcommand{\Eoperators}
{
\begin{align*}
E_{\X} &= \frac{1}{ 4 \q } \big[ 3 ( R + S ) - ( S R S + R S R ) \big],\\
E_{\Y} &= \frac{- i}{2} [R, S],\\
E_{\Z} &= \frac{1}{ 4 \q } \big[ 3 ( R - S ) - ( S R S - R S R ) \big].
\end{align*}
}
\newcommand{\avrh}[2]
{
\ave[#2]{#1, \rho_{AB}}
}
\newcommand{\fnrhh}[1]
{
\norm[\big]{ #1 \rho_{AB}^{1/2} }_{F}
}
\newcommand{\q}{\sqrt{2}}
\newcommand{\lemak}[3]
{
\begin{lem}
\label{#1}
If the observed violation is close to maximal, then #2. More specifically, if $\ave{ W, \rho_{AB} } = 6 - \varepsilon$, then
#3
\end{lem}
}
\newcommand{\AAanticomm}
{
\fnrhh{ \{ A_{0}, A_{1} \} \otimes \I \,} \leq 2 \big( 1 + \q \big) \sqrt{\varepsilon}
}
\newcommand{\AAAanticomm}
{
\fnrhh{ \{ A_{0} + A_{1}, A_{2} \} \otimes \I \,} \leq 2 \big( 4 + \q ) \sqrt{\varepsilon}
}
\begin{document}
\title{A weak form of self-testing}
\author{J\k{e}drzej Kaniewski}
\affiliation{Faculty of Physics, University of Warsaw, Pasteura 5, 02-093 Warsaw, Poland}
\email{jkaniewski@fuw.edu.pl}
\date{\today}
\begin{abstract}
The concept of self-testing (or rigidity) refers to the fact that for certain Bell inequalities the maximal violation can be achieved in an essentially unique manner. In this work we present a family of Bell inequalities which are maximally violated by multiple inequivalent quantum realisations. We completely characterise the quantum realisations achieving the maximal violation and we show that each of them requires a maximally entangled state of two qubits. This implies the existence of a new, weak form of self-testing in which the maximal violation allows us to identify the state, but does not fully determine the measurements. From the geometric point of view the set of probability points that saturate the quantum bound is a line segment. We then focus on a particular member of the family and show that the self-testing statement is robust, i.e.~that observing a non-maximal violation allows us to make a quantitative statement about the unknown state. To achieve this we present a new construction of extraction channels and analyse their performance. For completeness we provide two independent approaches: analytical and numerical. The noise robustness, i.e.~the amount of white noise at which the bound becomes trivial, of the analytical bound is rather small ($\approx 0.06\%$), but the numerical method takes us into an experimentally-relevant regime ($\approx 5\%$). We conclude by investigating the amount of randomness that can be certified using these Bell violations. Perhaps surprisingly, we find that the qualitative behaviour resembles the behaviour of rigid inequalities such as the Clauser--Horne--Shimony--Holt inequality. This shows that rigidity is not strictly necessary for device-independent applications.
\end{abstract}
\maketitle

\prlparagraph{Introduction}In his seminal work Bell showed that performing measurements on spatially-separated quantum systems may give rise to correlations inconsistent with any local-realistic description of the world~\cite{bell64a} (see Ref.~\cite{brunner14a} for a comprehensive review). While the initial motivation for studying Bell nonlocality and performing Bell experiments was to demonstrate, beyond any reasonable doubt, that the world is non-classical, we now understand that Bell nonlocality can be used in a more constructive manner. If we assume that the systems under consideration are governed by quantum mechanics, one can use the observed correlations to draw conclusions about their inner workings, a phenomenon known as \emph{device-independent certification} of quantum devices. Quite surprisingly, in some cases one can almost completely determine the state and measurements under consideration. First such statements can be traced back to the early works of Tsirelson~\cite{tsirelson87a, tsirelson93a}, Summers and Werner~\cite{summers87a} and Popescu and Rohrlich~\cite{popescu92a} and this phenomenon is now referred to as \emph{self-testing}~\cite{mayers98a, mayers04a} or \emph{rigidity}~\cite{reichardt13a}. The simplest and most well-known example concerns the famous Clauser--Horne--Shimony--Holt (CHSH) inequality~\cite{clauser69a}: if we observe the maximal violation of the CHSH inequality, we must be measuring a maximally entangled state of two qubits using anticommuting observables~\cite{tsirelson87a, summers87a, tsirelson93a, popescu92a}.

In a tomographic scenario we use a trusted measurement device to characterise an unknown quantum state (or vice versa) and in such a scenario a complete description of the unknown object can be obtained. In a device-independent scenario we trust neither the state nor the measurements, which imposes certain limitations on how much information we can hope to extract. As we have no information about the dimension of our system, we can never rule out the presence of additional degrees of freedom on which the measurements act trivially. Similarly, since there are no preferred local reference frames, we can only hope to characterise the system up to local unitaries. These two ambiguities are always present in the device-independent setting and any self-testing statement must account for them. A third ambiguity arises when the quantum realisation is \emph{chiral}, i.e.~it is not unitarily equivalent to its own transpose (we take the transpose in some fixed product basis). For instance the ordered set of three observables given by the Pauli matrices $( \X, \Y, \Z )$ is not unitarily equivalent to $( \X\tran, \Y\tran, \Z\tran )$. Several scenarios involving chiral realisations have been studied~\cite{mckague11a, kaniewski17a, andersson17a, kaniewski19b} and there the transpose ambiguity must be explicitly added to the list of allowed equivalences. Since we consider the transpose to be as natural and well-understood as the other two equivalences, we still refer to such a characterisation as self-testing.

By now several classes of self-testing statements have been derived~\cite{bardyn09a, mckague14a, mckague12a, yang13a, bamps15a, mckague16a, wang16a, supic16a, mckague17a, coladangelo17a, kalev17a, andersson17a, supic18a, coladangelo17c, kaniewski19b, sarkar19a} and all of them exhibit the same structure: observing some strongly non-classical correlations implies that particular local measurements are performed on a specific entangled state (up to the equivalences mentioned above). In some cases these statements have been made robust, which allows us to draw non-trivial conclusions in the presence of a realistic level of noise~\cite{bancal15a, yang14a, pal14a, wu16a, kaniewski16b, coopmans19a}. See Ref.~\cite{supic19a} for a recent review on self-testing.

In addition to its foundational importance self-testing has immediate applications to cryptography: if the Bell violation alone essentially determines the quantum realisation, one can certify that the randomness generated in the experiment is intrinsically quantum and cannot be known to an external eavesdropper. This is precisely the idea behind \emph{device-independent cryptography}~\cite{barrett05b, acin06a, colbeck06a, acin07a, pironio10a} (see Refs.~\cite{ekert14a, acin16b, bera17a} for reviews on various aspects of device-independent cryptography and randomness in quantum physics).

In this work we prove the existence of a new, weak form of self-testing. We study a 1-parameter family of Bell inequalities and show that observing the maximal violation certifies the presence of a maximally entangled state of two qubits even though the measurements cannot be uniquely determined. To understand how this phenomenon is affected by noise, we focus on a particular member of the family and derive an analytic robust self-testing result for the state. Since the analytic statement can only tolerate a small amount of noise, we also compute numerical bounds using the ``swap method''~\cite{bancal15a, yang14a}, which turn out to be significantly stronger. Finally, we study the amount of randomness that can be certified from the observed violation.

\prlparagraph{A family of Bell functionals}
Given a measurement with two outcomes $\{ F_{0}, F_{1} \}$ we associate the outcomes with values ${\pm 1}$, which gives rise to the observable $A := F_{0} - F_{1}$. We denote the observables of Alice and Bob by $A_{x}$ and $B_{y}$, respectively. In the bipartite scenario with three settings and two outcomes we consider a family of Bell functionals defined as
\begin{equation}
\label{eq:Bell-functional}
\begin{aligned}
\beta := &\ave{A_{0} B_{0}} + \ave{A_{0} B_{1}} + \alpha \ave{A_{0} B_{2}} + \ave{A_{1} B_{0}}\\
&+ \ave{A_{1} B_{1}} - \alpha \ave{A_{1} B_{2}} + \alpha \ave{A_{2} B_{0}} - \alpha \ave{A_{2} B_{1}},
\end{aligned}
\end{equation}
where $\alpha \in [0, 2]$ is a parameter and $\ave{A_{x} B_{y}}$ denotes the expectation value of the product of the outcomes. Note that for $\alpha = 1$ this is precisely the correlation part of the $I_{3322}$ Bell functional~\cite{froissart81a, collins04a}. It is easy to check that for this Bell functional the largest value achievable by local-realistic models equals $\beta_{L} = 4 \max \{1, \alpha \}$, whereas quantum systems can achieve the value of $\beta_{Q} = \qval$. For our purposes we are only interested in Bell functionals that satisfy $\beta_{L} < \beta_{Q}$ (Bell functionals satisfying $\beta_{L} = \beta_{Q}$ cannot be used to certify quantum properties as they do not allow us to rule out a local-realistic description of the system), so from now on we restrict our attention to the case of $\alpha \in (0, 2)$. It turns out that in those cases the quantum value can be achieved in multiple (inequivalent) ways: a 1-parameter family of quantum realisations for $\alpha = 1$ was presented in Ref.~\cite{goh18a} and can be straightforwardly generalised to all $\alpha \in (0, 2)$. This family is based on the maximally entangled state of two qubits $\ket{\Phi^{+}} = \frac{1}{\q} ( \ket{00} + \ket{11} )$ and gives rise to a line segment in the space of probability distributions. Hence, it serves as a simple example of a non-trivial (in the sense that $\beta_{L} < \beta_{Q}$) Bell functional that does not have a unique maximiser in the quantum set.

Until this work it was not known whether this Bell functional (1) admits additional maximisers that do not belong to the line segment and (2) exhibits some weak form of self-testing. In this work we answer both of these questions.

\prlparagraph{Exact self-testing}
Writing out the Bell operator gives
\begin{equation}
\label{eq:Bell-operator}
\begin{aligned}
W &= A_{0} \otimes ( B_{0} + B_{1} + \alpha B_{2} )\\
&+ A_{1} \otimes ( B_{0} + B_{1} - \alpha B_{2} ) + \alpha A_{2} \otimes ( B_{0} - B_{1} ).
\end{aligned}
\end{equation}
Since the Bell functional contains only correlators, the quantum value can be computed by solving a semidefinite program~\cite{wehner06a} and the dual solution can be turned into a sum-of-squares decomposition of the Bell operator~\cite{doherty08a}. Indeed, it is easy to verify that
\begin{align*}
2 W &= ( 2 A_{0}^{2} + 2 A_{1}^{2} + \alpha^{2} A_{2}^{2} ) \otimes \I + \I \otimes ( B_{0}^{2} + B_{1}^{2} + \alpha^{2} B_{2}^{2} )\\
&- \sum_{j = 0}^{2} L_{j}^{2},
\end{align*}
where
\begin{align*}
L_{0} &= ( A_{0} + A_{1} ) \otimes \I - \I \otimes ( B_{0} + B_{1} ),\\
L_{1} &= ( A_{0} - A_{1} ) \otimes \I - \alpha \I \otimes B_{2},\\
L_{2} &= \alpha A_{2} \otimes \I - \I \otimes ( B_{0} - B_{1} ).
\end{align*}
Since we do not a priori assume that the measurements are projective, we do not replace $A_{x}^{2}$ and $B_{y}^{2}$ by identity operators. Nevertheless, we still have $A_{x}^{2} \leq \I, B_{y}^{2} \leq \I$, which immediately implies that $W \leq (\qval) \, \I \otimes \I$. To see that this bound can be saturated consider the maximally entangled two-qubit state $\ket{\Phi^{+}}$ and the observables
\begin{align*}
A_{0} &= B_{0} = \cos \theta_{\alpha} \X + \sin \theta_{\alpha} \Z,\\
A_{1} &= B_{1} = \cos \theta_{\alpha} \X - \sin \theta_{\alpha} \Z,\\
A_{2} &= B_{2} = \Z,
\end{align*}
where $\theta_{\alpha} := \arcsin( \alpha / 2 )$ (note that the range $\alpha \in (0, 2)$ corresponds to $\theta_{\alpha} \in (0, \pi/2)$). Our goal now is to characterise all quantum realisations that achieve the maximal quantum value. To do so we use a method proposed originally in Ref.~\cite{popescu92a}, which proceeds in 4 steps: (1) find algebraic relations satisfied by the local observables, (2) explicitly characterise the local observables, (3) construct the Bell operator and (4) diagonalise it.

Let $\rho_{AB}$ be an arbitrary bipartite state $\rho_{AB}$ and $A_{x}, B_{y}$ be arbitrary binary observables. Since the form of the local observables outside of the reduced states $\rho_{A}$ and $\rho_{B}$ has no influence on the statistics, no statements can be made about how these operators act outside of the support of these local states. A convenient solution is to disregard the additional, unused dimensions by truncating the local Hilbert spaces until the reduced states are supported on the entire local space. Such a procedure ensures that the reduced states $\rho_{A}$ and $\rho_{B}$ are full-rank, which we assume throughout this work. If this realisation achieves $\avrh{ W }{} = \qval$, where $\ave{A, B} := \tr (A\hc B)$ is the Hilbert--Schmidt inner product, we deduce that
\begin{equation}
\label{eq:Ax-By-projective}
\ave{ A_{x}^{2}, \rho_{A} } = 1 \nbox{and} \ave{ B_{y}^{2}, \rho_{B} } = 1
\end{equation}
for $x, y \in \{0, 1, 2\}$ and moreover that
\begin{equation}
\label{eq:Vj-condition}
\avrh{ L_{j}^{2} }{} = 0
\end{equation}
for $j \in \{0, 1, 2\}$. Since the reduced states are full-rank Eq.~\eqref{eq:Ax-By-projective} implies that all the measurements are projective, i.e.~$A_{x}^{2} = \I$ and $B_{y}^{2} = \I$. Conditions given in Eq.~\eqref{eq:Vj-condition}, on the other hand, impose some constraints on how the observables of Alice and Bob act on the state. Since $\avrh{ L_{j}^{2} }{} = \fnrhh{ L_{j} }^{2}$, where $\norm{\cdot}_{F}$ is the Frobenius norm $\norm{A}_{F} := \sqrt{ \ave{ A, A } }$, the equality $\avrh{ L_{j}^{2} }{} = 0$ implies that the operator $L_{j} \rho_{AB}^{1/2}$ vanishes. As an immediate consequence we obtain $L_{j} \rho_{AB} = 0$ for all $j \in \{0, 1, 2\}$. These conditions involve the observables of both parties, but we can use projectivity deduced earlier to eliminate one of them. By elementary algebraic manipulations (see Appendix~\ref{app:exact-self-testing} for details) we show that $L_{1} \rho_{AB} = 0$ implies
\begin{equation}
\label{eq:A0-A1-anticommutator}
\{ A_{0}, A_{1} \} = ( 2 - \alpha^{2} ) \I.
\end{equation}
Furthermore, equality $(L_{0} + L_{2}) \rho_{AB} = 0$ implies
\begin{equation}
\label{eq:A0+A1-A2-anticommutator}
\{ A_{0} + A_{1}, A_{2} \} = 0.
\end{equation}
It is well-known that the commutation relation given in Eq.~\eqref{eq:A0-A1-anticommutator} combined with $A_{0}^{2} = A_{1}^{2} = \I$ implies a particular form of $A_{0}$ and $A_{1}$ (see Ref.~\cite{kaniewski17a} for an elementary proof). The Hilbert space of Alice must be of the form $\sH_{A} \equiv \amsbb{C}^{2} \otimes \amsbb{C}^{d_{A}}$ for some $d_{A} \in \amsbb{N}$ and up to a local unitary the observables can be written as
\begin{align*}
A_{0} &= ( \cos \theta_{\alpha} \X + \sin \theta_{\alpha} \Z ) \otimes \I,\\
A_{1} &= ( \cos \theta_{\alpha} \X - \sin \theta_{\alpha} \Z ) \otimes \I.
\end{align*}
Then, Eq.~\eqref{eq:A0+A1-A2-anticommutator} implies that
\begin{equation*}
A_{2} = \sum_{j = 1}^{d_{A}} \big( \cos u_{j} \Y + \sin u_{j} \Z \big) \otimes \ketbraq{ a_{j} },
\end{equation*}
where $u_{j} \in [0, 2 \pi)$ and $\{ \ket{a_{j}} \}_{j = 1}^{d_{A}}$ forms an orthonormal basis on $\amsbb{C}^{d_{A}}$. It is convenient to think of this arrangement of observables as a direct sum of $2 \times 2$ subspaces where each subspace is characterised by an angle $u_{j} \in [0, 2\pi)$.

Since the Bell functional is symmetric with respect to swapping Alice and Bob, the Hilbert space of Bob must also decompose as $\sH_{B} \equiv \amsbb{C}^{2} \otimes \amsbb{C}^{d_{B}}$ for some $d_{B} \in \amsbb{N}$ and the observables must be of the same form. However, it is convenient to write them down in a slightly different manner:
\begin{align*}
B_{0} &= \sum_{k = 1}^{d_{B}} \big[ \cos \theta_{\alpha} \X + \sin \theta_{\alpha} ( - \cos v_{k} \Y + \sin v_{k} \Z ) \big] \otimes \ketbraq{b_{k}},\\
B_{1} &= \sum_{k = 1}^{d_{B}} \big[ \cos \theta_{\alpha} \X - \sin \theta_{\alpha} ( - \cos v_{k} \Y + \sin v_{k} \Z ) \big] \otimes \ketbraq{b_{k}},\\
B_{2} &= \Z \otimes \I,
\end{align*}
where $v_{k} \in [0, 2 \pi)$ and $\{ \ket{b_{k}} \}_{k = 1}^{d_{B}}$ forms an orthonormal basis on $\amsbb{C}^{d_{B}}$.

Having characterised the local observables we are ready to write down the Bell operator. It is convenient to reorder the registers and write it as
\begin{equation*}
W = \sum_{j = 1}^{ d_{A} } \sum_{k = 1}^{ d_{B} } R( u_{j}, v_{k} ) \otimes \ketbraq{ a_{j} } \otimes \ketbraq{ b_{k} },
\end{equation*}
where $R(u, v)$ is the two-qubit Bell operator corresponding to angle $u$ for Alice and $v$ for Bob. To characterise the states which give rise to the maximal violation we must find out for which choices of $u$ and $v$ the value $\lambda = \qval$ is an eigenvalue of $W$ and what the corresponding eigenspace is. Since $\ketbraq{a_{j}} \otimes \ketbraq{b_{k}}$ are orthogonal projectors, the spectrum of $W$ is simply the union of the spectra of $R(u_{j}, v_{k})$. The two-qubit operator $R(u, v)$ can be diagonalised explicitly and the only eigenvalue that can attain the maximal value of $\qval$ is given by
\begin{equation*}
\lambda_{\textnormal{max}}( R(u, v) ) = 4 + \alpha^{2} \Big[ 2 \cos \Big( \frac{ u - v }{2} \Big) - 1 \Big],
\end{equation*}
where we have eliminated $\theta_{\alpha}$ using the relation $\alpha = 2 \sin \theta_{\alpha}$. The eigenvalue $\lambda = \qval$ appears iff $u = v$, the corresponding eigenspace is 1-dimensional and one can check that thanks to the particular choice of Bob's observables the corresponding eigenvector is always $\ket{\Phi^{+}}$. The fact that in this continuous family of two-qubit realisations parametrised by angle $u \in [0, 2\pi)$ the optimal state does not depend on the angle allows us to conclude that any state $\rho_{AB}$ satisfying $\avrh{W}{} = \qval$ must up to local unitaries be of the form:
\begin{equation*}
\rho_{AB} = \Phi^{+}_{A'B'} \otimes \sigma_{A''B''},
\end{equation*}
where $\sigma_{A''B''}$ is a normalised state satisfying
\begin{equation*}
\ave[\Big]{ \sigma_{A''B''}, \ketbraq{ a_{j} } \otimes \ketbraq{ b_{k} } } = 0
\end{equation*}
whenever $u_{j} \neq v_{k}$ (the state $\sigma_{A''B''}$ is only supported on the subspaces where Alice and Bob perform ``matching'' measurements). In other words, every quantum realisation that achieves the maximal quantum value is basically a convex combination of the two-qubit realisations presented above.

We are now able to characterise all the probability distributions which saturate the quantum value and it suffices to compute the statistics corresponding to the two-qubit realisations. Since the state is maximally entangled, we have $\ave{ A_{x} } = \ave{ B_{y} } = 0$, while the correlators are given by
\begin{align*}
\ave{A_{0} B_{0}} &= \ave{A_{1} B_{1}} = 1 - \frac{\alpha^{2}}{4} ( 1 - \sin u ),\\
\ave{A_{0} B_{1}} &= \ave{A_{1} B_{0}} = 1 - \frac{\alpha^{2}}{4} ( 1 + \sin u ),\\
\ave{A_{0} B_{2}} &= \ave{A_{2} B_{0}} = \frac{\alpha}{2},\\
\ave{A_{1} B_{2}} &= \ave{A_{2} B_{1}} = - \frac{\alpha}{2},\\
\ave{A_{2} B_{2}} &= \sin u.
\end{align*}
It is clear that this set, which is an exposed face of the quantum set of correlations, is simply a line segment and that the extremal points correspond to $u = \pi/2$ and $u = 3 \pi/2$. Note also that choosing $u = x$ and $u = \pi - x$ leads to identical statistics, because the corresponding realisations are related by a transpose.

It is natural to ask whether a stronger self-testing statement can be made if instead of looking at the Bell value we consider the entire statistics. It is easy to see that the extremal points of the line segment are self-tests in the usual sense, i.e.~the exact form of observables can be deduced (in fact, since they are extremal points of the quantum set of correlators, this follows already from the work of Tsirelson~\cite{tsirelson87a}). The points in the interior, on the other hand, cannot be self-tests in the usual sense, since they are not extremal in the quantum set. Moreover, it is easy to see that each interior point can be achieved in at least two inequivalent ways: (1) by a particular two-qubit realisation corresponding to a specific value of $u$ or (2) as a convex combination of the two extremal points. Nevertheless, all such points certify the maximally entangled state of two qubits.

\prlparagraph{Robust self-testing}
To study the case of non-maximal violation we focus on the Bell functional which corresponds to $\alpha = \q$. This is a convenient choice because in this case all the ideal realisations employ a pair of anticommuting observables. Here we present a robust self-testing of the observables and the state (see Appendices~\ref{app:extraction-channels} and~\ref{app:robust-self-testing} for derivations). For the following two theorems we assume that $W$ is a Bell operator obtained by setting $\alpha = \q$ in Eq.~\eqref{eq:Bell-operator}.

\begin{figure}[h]
\includegraphics{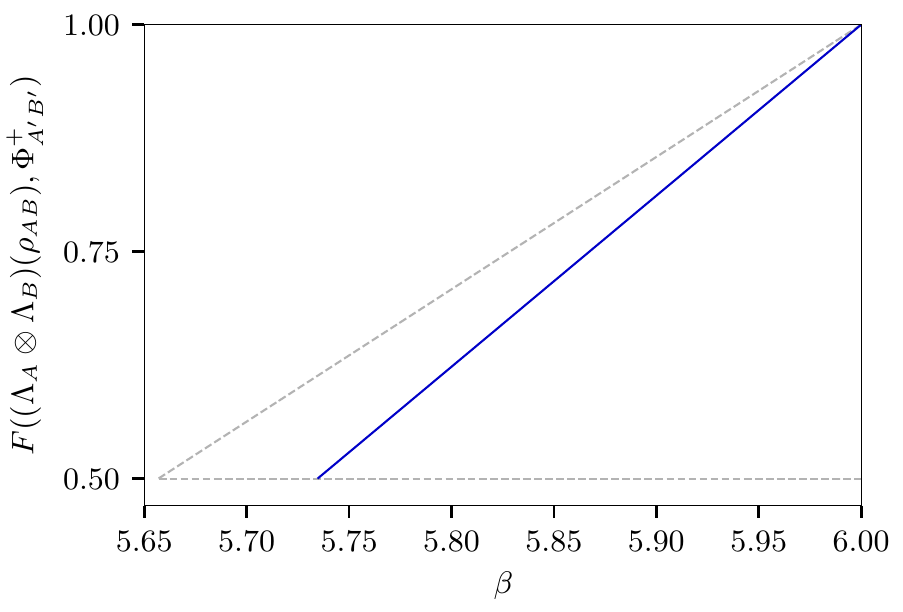}
\caption{The blue line represents a lower bound on the fidelity of the extracted two-qubit state computed using the swap method as a function of the observed violation. The gray lines correspond to the trivial upper and lower bounds.}
\label{fig:self-test}
\end{figure}

In the exact case we have concluded that $A_{x}^{2} = \I$, $\{A_{0}, A_{1}\} = 0$ and $\{ A_{0} + A_{1}, A_{2} \} = 0$. The sum-of-squares decomposition implies that these algebraic relations are approximately satisfied if a near-maximal violation is observed.
\begin{thm}
\label{thm:robust-self-testing-observables}
If $\ave{ W, \rho_{AB} } \geq 6 - \varepsilon$, then the measurements are nearly projective:
\begin{equation*}
\ave{ A_{x}^{2}, \rho_{A} }{} \geq 1 - \varepsilon
\end{equation*}
for $x \in \{0, 1, 2\}$. Moreover, the following pairs of operators approximately anticommute:
\begin{align*}
\ave{ \{ A_{0}, A_{1} \}^{2}, \rho_{A} } &\leq 4 ( 3 + 2 \q ) \varepsilon,\\
\ave{ \{A_{0} + A_{1}, A_{2} \}^{2}, \rho_{A} } &\leq 8 ( 9 + 4 \q ) \varepsilon.
\end{align*}
\end{thm}
By symmetry analogous statements hold for the observables of Bob. Note that these statements remain non-trivial even under a macroscopic amount of noise (e.g.~the trivial bound for the second quantity reads $\ave{ \{ A_{0}, A_{1} \}^{2}, \rho_{A} } \leq 4$, which is saturated by all projective measurements whose operators commute).

A complete characterisation of the optimal arrangements derived above allows us to propose suitable extraction channels and what is novel is the fact that one of the extraction channels must depend on all three observables.
\begin{thm}
\label{thm:robust-self-testing-state}
If $\ave{ W, \rho_{AB} } \geq 6 - \varepsilon$, then there exist local extraction channels $\Lambda_{A}$ and $\Lambda_{B}$ such that
\begin{equation*}
F( ( \Lambda_{A} \otimes \Lambda_{B} )( \rho_{AB} ) , \Phi^{+}_{A'B'} ) \geq 1 - \frac{1}{4} \big( 18 + 11\q \big) \sqrt{\varepsilon}.
\end{equation*}
\end{thm}
This statement is not particularly robust to noise: the right-hand side exceeds the trivial value of $\frac{1}{2}$ only if $\varepsilon \leq 0.0035$. To obtain stronger results we have employed the swap method and the results are presented in Fig.~\ref{fig:self-test}. The lower bound on the fidelity of the extracted state is essentially a straight line and strongly resembles the best currently known bound for the CHSH inequality (cf.~Fig.~1 in Ref.~\cite{kaniewski16b}).

\prlparagraph{Certifying randomness}
\begin{figure}[h]
\includegraphics{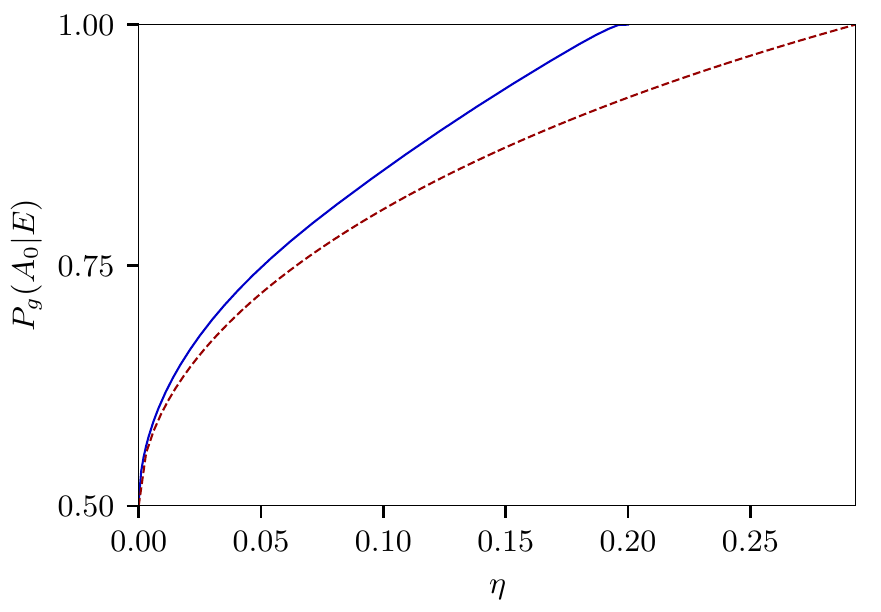}
\caption{Comparison of the randomness certification power of the new inequality (solid blue line) and the CHSH inequality (dashed red line). We plot upper bounds on the probability that Eve successfully guesses the outcome of the $A_{0}$ measurement as a function of the noise parameter $\eta$.}
\label{fig:chsh-comparison}
\end{figure}
We have so far focused solely on certifying quantum properties such as anticommutation of observables or the presence of a particular quantum state. The next natural question concerns the randomness that can be certified from the Bell violation against an external eavesdropper Eve. If we observe the maximal violation, we can draw conclusions from the complete characterisation of optimal quantum realisations derived above. Contrary to the usual scenario we can no longer argue that Eve is completely decoupled from the degrees of freedom on which the measurements of Alice and Bob act non-trivially. Nevertheless, if we only care about the randomness produced by a single observable of a single party, we can still guarantee maximal randomness, because all the optimal realisations involve rank-1 projective measurements acting on a qubit which is maximally entangled with the trusted party. To examine what happens in the presence of noise we have performed numerical calculations for the case $\alpha = 1$. Perhaps surprisingly, we have found that the qualitative behaviour resembles closely that of standard rigid inequalities. To make a fair comparison with the CHSH inequality suppose that in both cases the ideal measurements are performed on the isotropic two-qubit state $\sigma(\eta) := (1 - \eta) \Phi^{+} + \eta \I \otimes \I / 4$, where $\eta \in [0, 1]$ is a noise parameter. In Fig.~\ref{fig:chsh-comparison} we compare our numerical results with the well-known analytic trade-off for the CHSH inequality~\cite{pironio10a} (see Appendix~\ref{app:randomness} for details).

\prlparagraph{Conclusions and outlook}
Self-testing is an active research field and a particularly interesting direction is to explore its powers and limitations by deriving new types of self-testing statements or impossibility results. For instance we have recently learnt that one can self-test quantum channels~\cite{sekatski18a}, entangled measurements~\cite{bancal18a, renou18a}, quantum instruments~\cite{wagner20a} or that one can extend the concept of self-testing to prepare-and-measure scenarios~\cite{tavakoli18a, farkas19a, tavakoli20a, mironowicz19a, miklin20a, mohan19a}. In this work we derive a new type of self-testing statement which allows us to certify the state but not the measurements.

Until now self-testing of the state or randomness certification have only been shown for rigid Bell inequalities and so one might have conjectured rigidity to be necessary for these purposes. In this work we show that the non-rigid nature of a Bell inequality does not prevent it from being a robust self-test of a quantum state or an efficient certificate for randomness.

The first question that follows from our work is whether there exist applications in which rigidity is actually strictly necessary. Can we find a natural and operational task in which non-rigid inequalities exhibit a qualitatively different behaviour? A different direction would be to look for even weaker forms of self-testing. The Bell inequalities considered in this work do not certify the entire quantum realisation, but at least uniquely determine the state. We are not aware of any bipartite Bell inequalities which are maximally violated by multiple inequivalent states, but if they exist, could they be used to make some even weaker form of self-testing statements? More generally, can we think of other natural generalisations of the concept of self-testing and device-independent certification?

\textsl{Note added.} Recently we became aware of Ref.~\cite{jebarathinam19a}, which investigates how self-testing and the geometry of the quantum set are affected by liftings.

\section*{Acknowledgments}
The project ``Robust certification of quantum devices'' is carried out within the HOMING programme of the Foundation for Polish Science co-financed by the European Union under the European Regional Development Fund.
\onecolumngrid
\appendix
\section{Exact self-testing}
\label{app:exact-self-testing}
In the main text we have shown how to completely characterise arrangements of local observables that are capable of achieving the quantum value, but we have skipped some elementary steps. In this appendix we fill in the details of this argument.

Writing out $L_{1} \rho_{AB} = 0$ gives
\begin{equation*}
\big[ ( A_{0} - A_{1} ) \otimes \I \big] \rho_{AB} = \alpha ( \I \otimes B_{2} ) \rho_{AB}.
\end{equation*}
The fact that all the measurements are projective implies that
\begin{align*}
\alpha^{2} \rho_{AB} = \alpha^{2} ( \I \otimes B_{2}^{2} ) \rho_{AB} &= \alpha \big[ ( A_{0} - A_{1} ) \otimes B_{2} \big] \rho_{AB}\\
&= \big[ ( A_{0} - A_{1} )^{2} \otimes \I \big] \rho_{AB}.
\end{align*}
Tracing out the register of Bob gives
\begin{equation*}
\alpha^{2} \rho_{A} = ( A_{0} - A_{1} )^{2} \rho_{A}.
\end{equation*}
Since the reduced state $\rho_{A}$ is full-rank, we can right-multiply by $\rho_{A}^{-1}$ to obtain
\begin{equation*}
\alpha^{2} \I = ( A_{0} - A_{1} )^{2},
\end{equation*}
which can be rearranged to give
\begin{equation*}
\{ A_{0}, A_{1} \} = ( 2 - \alpha^{2} ) \I.
\end{equation*}
Similarly, writing out $(L_{0} + L_{2}) \rho_{AB} = 0$ gives
\begin{equation*}
\big[ ( A_{0} + A_{1} + \alpha A_{2} ) \otimes \I \big] \rho_{AB} = 2 ( \I \otimes B_{0} ) \rho_{AB},
\end{equation*}
which through an analogous argument leads to
\begin{equation*}
( A_{0} + A_{1} + \alpha A_{2} )^{2} = 4 \, \I.
\end{equation*}
Combining this with the relation derived above gives
\begin{equation*}
\{ A_{0} + A_{1}, A_{2} \} = 0.
\end{equation*}
We now choose the basis such that $A_{0}$ and $A_{1}$ are given by
\begin{align*}
A_{0} &= ( \cos \theta_{\alpha} \X + \sin \theta_{\alpha} \Z ) \otimes \I,\\
A_{1} &= ( \cos \theta_{\alpha} \X - \sin \theta_{\alpha} \Z ) \otimes \I.
\end{align*}
To find all valid solutions for $A_{2}$ we start by writing $A_{2}$ as
\begin{equation*}
A_{2} = \I \otimes T_{\I} + \X \otimes T_{\X} + \Y \otimes T_{\Y} + \Z \otimes T_{\Z}
\end{equation*}
for some Hermitian operators $T_{\I}, T_{\X}, T_{\Y}, T_{\Z}$ acting on $\amsbb{C}^{d_{A}}$. Equality $\{ A_{0} + A_{1}, A_{2} \} = 0$ immediately implies that $T_{\I} = T_{\X} = 0$. It is then easy to check that
\begin{equation*}
A_{2}^{2} = \I \otimes ( T_{\Y}^{2} + T_{\Z}^{2} ) + i \X \otimes [ T_{\Y}, T_{\Z} ].
\end{equation*}
The condition $A_{2}^{2} = \I$ implies that
\begin{equation*}
T_{\Y}^{2} + T_{\Z}^{2} = \I \nbox{and} [T_{\Y}, T_{\Z}] = 0.
\end{equation*}
Since $T_{\Y}$ and $T_{\Z}$ commute, there exists a basis in which they are both diagonal and let us denote such a basis by $\{ \ket{ a_{j} } \}_{j = 1}^{d_{A}}$. The first condition implies that the eigenvalues of $T_{\Y}$ and $T_{\Z}$ can be expressed as $\cos u_{j}$ and $\sin u_{j}$ of some angle $u_{j} \in [0, 2 \pi)$ and therefore
\begin{align*}
T_{\Y} &= \sum_{j = 1}^{d_{A}} \cos u_{j} \, \ketbraq{ e_{j} },\\
T_{\Z} &= \sum_{j = 1}^{d_{A}} \sin u_{j} \, \ketbraq{ e_{j} }
\end{align*}
This immediately implies that
\begin{equation*}
A_{2} = \sum_{j = 1}^{d_{A}} \big( \cos u_{j} \, \Y + \sin u_{j} \, \Z \big) \otimes \ketbraq{ e_{j} },
\end{equation*}
which is precisely the form given in the main text.
\section{Extraction channels}
\label{app:extraction-channels}
In this appendix we propose two explicit constructions of extraction channels tailored to the case of $\alpha = \q$. The first one is an extension of the standard swap isometry given in Ref.~\cite{mckague12a}, while the second one is a novel construction. The reason why new constructions are necessary is the fact that at least one of the extraction channels must depend on all three observables. At this point we are only interested in certifying the state, so we can without loss of generality assume that the measurements of Alice and Bob are projective (mapping non-projective measurements onto projective ones can always be seen as the first part of the extraction process).
\subsection{Preliminaries}
Let us start by proving two simple facts about binary observables. Both proofs rely crucially on Jordan's lemma, which states that two Hermitian operators satisfying $R^{2} = S^{2} = \I$ can be simultaneously block-diagonalised such that the resulting blocks are of size at most $2 \times 2$.
\begin{lem}
\label{lem:Fx-norm}
Let $R, S$ be Hermitian operators acting on $\amsbb{C}^{d}$ satisfying $R^{2} = S^{2} = \I$. Then, the operator
\begin{equation*}
T := \frac{1}{ 4 \q } \big[ 3 ( R + S ) - ( S R S + R S R ) \big]
\end{equation*}
satisfies $- \I \leq T \leq \I$.
\end{lem}
\begin{proof}
Thanks to Jordan's lemma it suffices to consider observables acting on $\amsbb{C}^{2}$. Up to unitaries these can be parametrised as
\begin{equation}
\begin{aligned}
\label{eq:two-observables}
R &= \cos \theta \, \X + \sin \theta \, \Z,\\
S &= \cos \theta \, \X - \sin \theta \, \Z,
\end{aligned}
\end{equation}
for $\theta \in [0, \pi/2]$. For these operators a direct calculation shows that
\begin{equation*}
T = \frac{ 3 \cos \theta - \cos 3 \theta }{ 2 \q } \, \X.
\end{equation*}
Now it suffices to check that $\abs{ 3 \cos \theta - \cos 3 \theta } \leq 2 \q$ for all $\theta$.
\end{proof}
Let $\cL(\amsbb{C}^{d})$ be the set of linear operators acting on $\amsbb{C}^{d}$.
\begin{lem}
\label{lem:LambdaB-positivity}
Let $R, S$ be Hermitian operators acting on $\amsbb{C}^{d}$ satisfying $R^{2} = S^{2} = \I$. Then, the linear map $\Lambda_{B} : \cL( \amsbb{C}^{d} ) \to \cL ( \amsbb{C}^{2} )$ defined as
\LambdaB{}
where
\Eoperators
is completely positive.
\end{lem}
\begin{proof}
To show that $\Lambda_{B}$ is completely positive we compute the corresponding Choi operator and prove that it is positive semidefinite. The unnormalised Choi operator is defined as
\begin{equation*}
C := ( \id_{A} \otimes \Lambda_{B} ) ( \ketbraq{\Omega}_{AB} ),
\end{equation*}
where $\ket{\Omega}_{AB} = \sum_{j = 1}^{d} \ket{j}_{A} \ket{j}_{B}$ is the standard (unnormalised) maximally entangled state of local dimension $d$. An explicit calculation gives
%
\begin{equation*}
C = \I \otimes \I + E_{\X}^{*} \otimes \X + E_{\Y}^{*} \otimes \Y + E_{\Z}^{*} \otimes \Z.
\end{equation*}
Since taking a (total) transpose does not affect the eigenvalues, it suffices to prove that $C\tran \geq 0$ and note that
\begin{equation*}
C\tran = \I \otimes \I + E_{\X} \otimes \X - E_{\Y} \otimes \Y + E_{\Z} \otimes \Z,
\end{equation*}
because the operators $E_{\X}, E_{\Y}, E_{\Z}$ are Hermitian.

The fact that $R$ and $S$ can be written in a block-diagonal form where the blocks are of size at most $2 \times 2$ implies that the same property holds for $E_{\X}, E_{\Y}$ and $E_{\Z}$. This means that to ensure that $C\tran \geq 0$, it suffices to check positivity for all possible observables in $d = 2$. Using the parametrisation given in Eq.~\eqref{eq:two-observables} we obtain
\begin{equation*}
C\tran = \I \otimes \I + \frac{ 3 \cos \theta - \cos 3 \theta }{2 \q} \, \X \otimes \X - \sin 2 \theta \, \Y \otimes \Y + \frac{ 3 \sin \theta + \sin 3 \theta }{2 \q} \, \Z \otimes \Z.
\end{equation*}
Clearly, this operator is diagonal in the Bell basis and the eigenvalues can be computed analytically. It is a simple exercise to check that the resulting trigonometric functions are non-negative on the interval $\theta \in [0, \pi/2]$.
\end{proof}
\subsection{Constructing an extraction channel from two observables}
Here we present two distinct ways of constructing a qubit extraction channel out of two binary observables acting on an unknown Hilbert space. Let $R$ and $S$ be binary observables corresponding to projective measurements, i.e.~Hermitian operators acting on $\amsbb{C}^{d}$ satisfying $R^{2} = S^{2} = \I$. It is well-known that if the observables anticommute $\{ R, S \} = 0$, they identify a qubit within $\amsbb{C}^{d}$. Our goal is to find simple constructions of linear maps that give rise to valid quantum channels for all choices of $R$ and $S$, while for observables satisfying $\{R, S\} = 0$ extract the desired qubit.

\textbf{Construction A.} The standard swap isometry is defined through the following circuit:
\begin{center}
\begin{tikzpicture}[scale=1, line width=0.8]
\node at (-4.75, 0) {$( \ket{0} + \ket{1} ) /\q$};
\draw (-3.5, 0) -- (3.5, 0);
\draw[fill=white] (-0.5, -0.5) rectangle (0.5, 0.5);
\node at (0, 0) {$H$};
\draw[fill=black] (-2, 0) circle (0.1);
\draw[fill=black] (2, 0) circle (0.1);
\node at (4, 0) {$\amsbb{C}^{2}$};
\draw (-2, 0) -- (-2, -2);
\draw (2, 0) -- (2, -2);
\node at (-4, -2) {$\rho$};
\draw (-3.5, -2) -- (3.5, -2);
\draw[fill=white] (-2.5, -2.5) rectangle (-1.5, -1.5);
\node at (-2, -2) {$R$};
\draw[fill=white] (1.5, -2.5) rectangle (2.5, -1.5);
\node at (2, -2) {$S$};
\node at (4, -2) {$\amsbb{C}^{d}$};
%
%
\end{tikzpicture}
\end{center}
The circuit corresponds to a concatenation of an isometry $V_{1} : \amsbb{C}^{d} \to \amsbb{C}^{2} \otimes \amsbb{C}^{d}$ and two unitaries $V_{2}, V_{3} : \amsbb{C}^{2} \otimes \amsbb{C}^{d} \to \amsbb{C}^{2} \otimes \amsbb{C}^{d}$ defined as
\begin{align*}
V_{1} &:= \frac{1}{\q} \big( \ket{0} \otimes \I + \ket{1} \otimes R \big),\\
V_{2} &:= H \otimes \I,\\
V_{3} &:= \ketbraq{0} \otimes \I + \ketbraq{1} \otimes S,
\end{align*}
where $H$ is the Hadamard matrix. It is easy to check that
%
\begin{align*}
%
%
V_{1} \ket{\psi} &= \frac{1}{\q} \ket{0} \ket{\psi} + \frac{1}{\q} \ket{1} R \ket{\psi},\\
V_{2} V_{1} \ket{\psi} &= \frac{1}{2} \ket{0} ( \I + R ) \ket{\psi} + \frac{1}{2} \ket{1} ( \I - R ) \ket{\psi},\\
V_{3} V_{2} V_{1} \ket{\psi} &= \frac{1}{2} \ket{0} ( \I + R ) \ket{\psi} + \frac{1}{2} \ket{1} S ( \I - R ) \ket{\psi}.
\end{align*}
The combined isometry $V : \amsbb{C}^{d} \to \amsbb{C}^{2} \otimes \amsbb{C}^{d}$ is given by $V := V_{3} V_{2} V_{1}$ and a direct computation shows that
\begin{align*}
V \rho V\hc = &\frac{1}{4} \Big[ \ketbraq{0} \otimes ( \I + R ) \rho ( \I + R ) + \ketbra{0}{1} \otimes ( \I + R ) \rho ( \I - R ) S\\
&+ \ketbra{1}{0} \otimes S ( \I - R ) \rho ( \I + R ) + \ketbraq{1} \otimes S ( \I - R ) \rho ( \I - R ) S \Big].
\end{align*}
Let $\Lambda_{A} : \cL( \amsbb{C}^{d} ) \to \cL( \amsbb{C}^{2} )$ be the quantum channel obtained by first applying the isometry and then tracing out the second register:
\begin{equation*}
\Lambda_{A}(\rho) := \tr_{2}( V \rho V\hc ).
\end{equation*}
Writing the output of the channel in the Pauli basis gives
\begin{equation}
\label{eq:lambdaA}
\Lambda_{A}(\rho) = \frac{1}{2} \ave{ \I, \rho } \, \I + \frac{1}{4} \ave{ S - R S R, \rho} \, \X + \frac{i}{4} \ave{ [ R, S ], \rho} \, \Y + \frac{1}{2} \ave{ R, \rho } \, \Z.
\end{equation}
If the observables anticommute $\{ R, S \} = 0$, it is easy to see that the $\X$ component of the output qubit is perfectly correlated to the $S$ observable on the initial system, while the $\Z$ component is perfectly correlated to the $R$ observable.

For our purposes we need to generalise this construction. Suppose that the operator $S$ instead of satisfying $S^{2} = \I$ is only guaranteed to satisfy $S^{2} \leq \I$. Since $\I - S^{2} \geq 0$, we can find a Hermitian operator $T$ satisfying $T^{2} = \I - S^{2}$. Then, consider
\begin{equation*}
\Phi_{3}(\rho) := \sum_{j = 0}^{1} K_{j} \rho K_{j}\hc,
\end{equation*}
where the Kraus operators are given by
\begin{align*}
K_{0} &:= \ketbraq{0} \otimes \I + \ketbraq{1} \otimes S,\\
K_{1} &:= \ketbraq{1} \otimes T.
\end{align*}
Clearly, this is a valid quantum channel. Let us now consider a swap circuit in which the unitary $V_{3}$ is replaced with the channel $\Phi_{3}$. It turns out that the resulting extraction channel is given precisely by Eq.~\eqref{eq:lambdaA}. In other words, this mathematical expression corresponds to a valid quantum channel for any $S$ satisfying $S^{2} \leq \I$.

\textbf{Construction B.} Consider a linear map $\Lambda_{B} : \cL( \amsbb{C}^{d} ) \to \cL( \amsbb{C}^{2} )$ defined as
\LambdaB{\frac{1}{2}}
where
\Eoperators
This map is clearly trace preserving, while complete positivity has been proved in Lemma~\ref{lem:LambdaB-positivity}.

This construction differs from the previous one in the sense that if $\{ R, S \} = 0$, then the $\X$ component of the output qubit is maximally correlated to $(R + S) / \q$, while the $\Z$ component is maximally correlated to $(R - S) / \q$.
\subsection{Combining the two channels}
In the previous section we have given two constructions of extraction channels and let us now explain how they can be applied to our self-testing scenario.

An essential requirement is that the extraction channels produce a perfect maximally entangled state of two qubits whenever the violation is maximal. Our explicit characterisation of the optimal strategies implies that when the maximal violation is achieved we have
\begin{align*}
\ave[\Big]{ &A_{2} \otimes \frac{ B_{0} - B_{1} }{\q}, \rho_{AB} } = 1,\\
\ave[\Big]{ &\frac{ A_{0} + A_{1} }{\q} \otimes \frac{ B_{0} +	 B_{1} }{\q}, \rho_{AB} } = 1.
\end{align*}
Now if we recall how the $\X$ and $\Z$ components of the output qubit are correlated to the observables $R$ and $S$ of the input system in the two constructions we arrive at the following choice of extraction channels. Alice employs the channel $\Lambda_{A}$ corresponding to $R = A_{2}$ and
\begin{equation*}
S = \frac{1}{ 4 \q } \big[ 3 ( A_{0} + A_{1} ) - ( A_{1} A_{0} A_{1} + A_{0} A_{1} A_{0} ) \big].
\end{equation*}
The fact that $S^{2} \leq \I$ follows immediately from Lemma~\ref{lem:Fx-norm}.
%
%
%
%
%
At the same time Bob employs the channel $\Lambda_{B}$ with $R = B_{0}$ and $S = B_{1}$. Let us denote the output two-qubit state by
\begin{equation*}
\sigma_{A'B'} := (\Lambda_{A} \otimes \Lambda_{B}) ( \rho_{AB} )
\end{equation*}
and our goal is to evaluate the fidelity between $\sigma_{A'B'}$ and the standard maximally entangled state $\Phi^{+}$. Since $\Phi^{+}$ is a pure state, we have
\begin{equation*}
F( \sigma_{A'B'}, \Phi^{+} ) = \ave{ \sigma_{A'B'}, \Phi^{+} }.
\end{equation*}
It is convenient to write $\Phi^{+}$ in the basis of Pauli matrices and evaluate each term separately. A direct calculation shows that for $\mathsf{P} \in \{ \X, \Y, \Z \}$ we have
\begin{equation*}
\ave{ \sigma_{A'B'} , \mathsf{P} \otimes \mathsf{P} } = \avrh{ C_{\mathsf{P}} }{},
\end{equation*}
where
%
%
%
%
\begin{align*}
&C_{\X} := \frac{1}{64} \big[ 3 ( A_{0} + A_{1} ) - ( A_{1} A_{0} A_{1} + A_{0} A_{1} A_{0} ) - 3 A_{2} ( A_{0} + A_{1} ) A_{2} + A_{2} ( A_{1} A_{0} A_{1} + A_{0} A_{1} A_{0} ) A_{2} \big]\\
&\qquad \otimes \big[ 3 ( B_{0} + B_{1} ) - ( B_{1} B_{0} B_{1} + B_{0} B_{1} B_{0} ) \big],\\
&C_{\Y} := \frac{- 1}{ 16 \q } \, \big[ A_{2}, 3 ( A_{0} + A_{1} ) - ( A_{1} A_{0} A_{1} + A_{0} A_{1} A_{0} ) \big] \otimes [B_{0}, B_{1}],\\
&C_{\Z} := \frac{1}{4 \q} \, A_{2} \otimes \big[ 3 ( B_{0} - B_{1} ) - ( B_{1} B_{0} B_{1} - B_{0} B_{1} B_{0} ) \big].
\end{align*}
%
%
%
%
%
Conveniently, it is not necessary to provide bounds on all three terms, because for every two-qubit state $\tau_{A'B'}$ we have
\begin{equation}
\label{eq:cross-terms-sum}
\ave{ \tau_{A'B'} , \X \otimes \X } + \ave{ \tau_{A'B'} , \Y \otimes \Y } + \ave{ \tau_{A'B'} , \Z \otimes \Z } \leq 1.
\end{equation}
To see this note that applying a correlated Pauli twirl to $\tau_{A'B'}$ produces a Bell-diagonal state without affecting the coefficients of the terms $\X \otimes \X$, $\Y \otimes \Y$ and $\Z \otimes \Z$ (see Lemma 10 in the supplementary information of Ref.~\cite{pfister16a} for more details). Positivity of the resulting density matrix immediately implies the condition given in Eq.~\eqref{eq:cross-terms-sum}. This means that
\begin{equation}
\label{eq:fidelity-lower-bound}
\begin{aligned}
F( \sigma_{A'B'}, \Phi^{+} ) &= \frac{1}{4} \big( 1 + \ave{ \sigma_{A'B'}, \X \otimes \X } - \ave{ \sigma_{A'B'}, \Y \otimes \Y } + \ave{ \sigma_{A'B'}, \Z \otimes \Z } \big)\\
&\geq \frac{1}{2} \big( \ave{ \sigma_{A'B'}, \X \otimes \X } + \ave{ \sigma_{A'B'}, \Z \otimes \Z } \big) = \frac{1}{2} \big( \avrh{ C_{\X} }{} + \avrh{ C_{\Z} }{} \big).
\end{aligned}
\end{equation}
In Appendix~\ref{app:analytic-bound} we derive analytic lower bounds on $\avrh{ C_{\X} }{}$ and $\avrh{ C_{\Z} }{}$ in terms of the observed violation, which lead to Theorem~2 in the main text.

Analogously, for every two-qubit state we have
\begin{equation}
- \ave{ \tau_{A'B'} , \X \otimes \X } - \ave{ \tau_{A'B'} , \Y \otimes \Y } + \ave{ \tau_{A'B'} , \Z \otimes \Z } \leq 1,
\end{equation}
which implies that
\begin{equation}
\label{eq:fidelity-lower-bound2}
F( \sigma_{A'B'}, \Phi^{+} ) \geq \frac{1}{2} \big( - \ave{ \sigma_{A'B'}, \Y \otimes \Y } + \ave{ \sigma_{A'B'}, \Z \otimes \Z } \big) = \frac{1}{2} \big( - \avrh{ C_{\Y} }{} + \avrh{ C_{\Z} }{} \big).
\end{equation}
This bound turns out to be more useful for the numerical calculations using the swap method given in Appendix~\ref{app:swap-method}.
\section{Robust self-testing}
\label{app:robust-self-testing}
In this appendix we derive robust self-testing bounds for the case of $\alpha = \q$. In the first part we derive analytic statements, whereas at the end we give some details on the numerical calculations performed using the swap method.
\subsection{Preliminaries}
Our main task is to bound norms of certain operators. We denote the Frobenius norm (Schatten 2-norm) by $\norm{\cdot}_{F}$ and the operator norm (Schatten $\infty$-norm) by $\norm{\cdot}_{\infty}$. Let us first state a couple of facts that we will take advantage of in the argument.

The Cauchy--Schwarz inequality for linear operators $X$ and $Y$ reads
\begin{equation}
\label{eq:cauchy-schwarz-inequality}
\abs{ \ave{ X, Y } } \leq \norm{X}_{F} \cdot \norm{Y}_{F}.
\end{equation}
We will often use this inequality in situations where one of the operators is a normalised quantum state. Note that then we have $\ave{L, \rho} = \ave{L \rho^{1/2}, \rho^{1/2}}$, which implies
\begin{equation}
\label{eq:cauchy-schwarz-spec}
\abs{ \ave{L, \rho} } \leq \norm[\big]{ L \rho^{1/2} }_{F}.
\end{equation}
Moreover, we will use the fact that
\begin{equation}
\label{eq:norm-inequality}
\norm{X Y}_{F} \leq \norm{X}_{F} \cdot \norm{Y}_{\infty}.
\end{equation}
This can be easily seen from the fact that
\begin{equation*}
\norm{X Y}_{F}^{2} = \tr ( X Y Y\hc X\hc ) \leq \norm{Y}_{\infty}^{2} \tr ( X X\hc ) = \norm{Y}_{\infty}^{2} \cdot \norm{X}_{F}^{2},
\end{equation*}
where we have used the fact that $Y Y\hc \leq \norm{Y}_{\infty}^{2} \, \I$ and that $A \geq B$ implies $\tr A  \geq \tr B$. We will also use the reverse triangle inequality which states that for any norm we have
\begin{equation}
\label{eq:reverse-triangle-inequality}
\abs[\big]{ \norm{X} - \norm{Y} } \leq \norm{X - Y}.
\end{equation}
Moreover, if $X^{2} = \I$, then
\begin{equation}
\label{eq:flip-identity}
(Y + X Y X)^{2} = \{ X, Y \}^{2}.
\end{equation}
\subsection{Conditions from the sum-of-squares decomposition}
Recall that for $\alpha = \q$ we have
\begin{equation*}
W = ( A_{0}^{2} + A_{1}^{2} + A_{2}^{2} ) \otimes \I + \I \otimes ( B_{0}^{2} + B_{1}^{2} + B_{2}^{2} ) - \frac{1}{2} \sum_{j = 0}^{2} L_{j}^{2}.
\end{equation*}
Clearly, if the observed violation equals $\beta = \avrh{W}{} = 6 - \varepsilon$, we can immediately deduce that
\begin{equation*}
\ave{ A_{x}^{2}, \rho_{A} } \geq 1 - \varepsilon
\end{equation*}
and
\begin{equation*}
\sum_{j = 0}^{2} \ave{ L_{j}^{2}, \rho_{AB} } \leq 2 \varepsilon.
\end{equation*}
The latter implies that
\begin{equation}
\label{eq:Vj-Frobenius}
\fnrhh{ L_{j} } = \sqrt{ \avrh{ L_{j}^{2} }{} } \leq \sqrt{ 2 \varepsilon }
\end{equation}
for $j = 0, 1, 2$.
\subsection{Analytic self-testing bounds}
\label{app:analytic-bound}
In this section we derive several robust self-testing statements. The techniques are elementary, but the proofs can be lengthy. To improve the readability we have divided the argument up into several lemmas.
\lemak{lem:A0-A1-commute}{the observables $A_{0}$ and $A_{1}$ approximately anticommute and, moreover, the operators $(A_{0} - A_{1})$ and $B_{2}$ are almost perfectly correlated}
{
\begin{equation*}
\AAanticomm
\end{equation*}
and
\begin{equation*}
\avrh{ ( A_{0} - A_{1} ) \otimes B_{2} }{} \geq \q - \sqrt{ 2 \varepsilon }.
\end{equation*}
}
\begin{proof}
Equation~\eqref{eq:Vj-Frobenius} applied to $L_{1}$ implies that
\begin{equation*}
\fnrhh{ [ ( A_{0} - A_{1} ) \otimes \I - \q \, \I \otimes B_{2} ] \,} \leq \sqrt{2 \varepsilon}.
\end{equation*}
If we multiply the operator under the norm by $\q \, \I \otimes B_{2}$ and then apply Eq.~\eqref{eq:norm-inequality} we conclude that
\begin{equation}
\label{eq:A0-A1}
\fnrhh{ [ \q ( A_{0} - A_{1} ) \otimes B_{2} - 2 \, \I \otimes \I ] \,} \leq 2 \sqrt{\varepsilon}.
\end{equation}
Alternatively, if we multiply the same operator by $( A_{0} - A_{1} ) \otimes \I$, we obtain
\begin{equation*}
\fnrhh{ [ 2 \, \I \otimes \I - \{A_{0}, A_{1} \} \otimes \I - \q ( A_{0} - A_{1} ) \otimes B_{2} ] \,} \leq 2 \sqrt{ 2 \varepsilon }.
\end{equation*}
These two inequalities allow us to apply the reverse triangle inequality to
\begin{align*}
X &= \{A_{0}, A_{1} \} \otimes \I \, \rho_{AB}^{1/2},\\
Y &= [ 2 \, \I \otimes \I - \q ( A_{0} - A_{1} ) \otimes B_{2} ] \, \rho_{AB}^{1/2},
\end{align*}
which gives the first inequality stated in the lemma. Inequality~\eqref{eq:A0-A1} together with the variant of the Cauchy--Schwarz inequality stated in~Eq.~\eqref{eq:cauchy-schwarz-spec} gives the second inequality stated in the lemma.
\end{proof}
\begin{cor}
\label{cor:B0-B1-commute}
Since the Bell inequality is symmetric with respect to swapping Alice and Bob, we immediately deduce that if $\ave{ W, \rho_{AB} } = 6 - \varepsilon$, then
\begin{equation*}
\fnrhh{ \I \otimes \{B_{0}, B_{1} \} \,} \leq 2 \big( 1 + \q \big) \sqrt{\varepsilon}
\end{equation*}
and
\begin{equation*}
\avrh{ A_{2} \otimes ( B_{0} - B_{1} ) }{} \geq \q - \sqrt{ 2 \varepsilon }.
\end{equation*}
\end{cor}
\lemak{lem:(A0+A1)x(B0+B1)}{the operators $(A_{0} + A_{1})$ and $(B_{0} + B_{1})$ are almost perfectly correlated}
{
\begin{equation*}
\avrh{ ( A_{0} + A_{1} ) \otimes ( B_{0} + B_{1} ) }{} \geq 2 - 2 \big( 1 + 2 \q \big) \sqrt{\varepsilon}.
\end{equation*}
}
\begin{proof}
Equation~\eqref{eq:Vj-Frobenius} applied to $L_{0}$ implies that
\begin{equation*}
\fnrhh{ [ ( A_{0} + A_{1} ) \otimes \I - \I \otimes ( B_{0} + B_{1} ) ] } \leq \sqrt{ 2 \varepsilon }.
\end{equation*}
If we multiply the operator under the norm by $\I \otimes ( B_{0} + B_{1} )$ and then apply Eq.~\eqref{eq:norm-inequality} we conclude that
\begin{equation*}
\fnrhh{ \big[ ( A_{0} + A_{1} ) \otimes ( B_{0} + B_{1} ) - 2 \, \I \otimes \I - \I \otimes \{ B_{0}, B_{1} \} \big] } \leq 2 \sqrt{ 2 \varepsilon }.
\end{equation*}
This together with Corollary~\ref{cor:B0-B1-commute} allow us to apply the reverse triangle inequality to
\begin{align*}
X &= \big[ ( A_{0} + A_{1} ) \otimes ( B_{0} + B_{1} ) - 2 \, \I \otimes \I \big] \rho_{AB}^{1/2},\\
Y &= \I \otimes \{B_{0}, B_{1} \} \, \rho_{AB}^{1/2},
\end{align*}
which gives
\begin{equation*}
\fnrhh{ \big[ ( A_{0} + A_{1} ) \otimes ( B_{0} + B_{1} ) - 2 \, \I \otimes \I \big] } \leq 2 \big( 1 + 2 \q \big) \sqrt{ \varepsilon }.
\end{equation*}
Applying the Cauchy--Schwarz inequality given in Eq.~\eqref{eq:cauchy-schwarz-spec} concludes the proof.
\end{proof}
\lemak{lem:A0+A1-A2-commute}{the operators $(A_{0} + A_{1})$ and $A_{2}$ approximately anticommute}
{
\begin{equation*}
\AAAanticomm.
\end{equation*}
}
\begin{proof}
Note that
\begin{equation*}
\fnrhh{ \big[ ( A_{0} + A_{1} + \q A_{2} ) \otimes \I - 2 \, \I \otimes B_{0} \big] } = \fnrhh{ ( L_{0} + L_{2} ) } \leq \fnrhh{ L_{0} } + \fnrhh{ L_{2} } \leq 2 \sqrt{ 2 \varepsilon }.
\end{equation*}
Multiplying the operator under the norm by $2 B_{0}$ gives
\begin{equation*}
\fnrhh{ \big[ 2 ( A_{0} + A_{1} + \q A_{2} ) \otimes B_{0} - 4 \, \I \otimes \I \big] } \leq 4 \sqrt{ 2 \varepsilon }.
\end{equation*}
Alternatively, multiplying it by $(A_{0} + A_{1} + \q A_{2} )$ gives
\begin{equation*}
\fnrhh{ \big[ ( A_{0} + A_{1} + \q A_{2} )^{2} \otimes \I - 2 \, ( A_{0} + A_{1} + \q A_{2} ) \otimes B_{0} \big] } \leq 2 ( 1 + \q ) \sqrt{\varepsilon}.
\end{equation*}
Since
\begin{equation*}
( A_{0} + A_{1} + \q A_{2} )^{2} = 4 \, \I + \{ A_{0}, A_{1} \} + \q \{ A_{0} + A_{1}, A_{2} \},
\end{equation*}
we can apply the reverse triangle inequality to
\begin{align*}
X &= \big( \{A_{0}, A_{1} \} + \q \{ A_{0} + A_{1}, A_{2} \} \big) \otimes \I \, \rho_{AB}^{1/2},\\
Y &= [ 2 ( A_{0} + A_{1} + \q A_{2} ) \otimes B_{0} - 4 \, \I \otimes \I ] \rho_{AB}^{1/2}
\end{align*}
to obtain
\begin{equation*}
\fnrhh{ \big( \{A_{0}, A_{1} \} + \q \{ A_{0} + A_{1}, A_{2} \} \big) \otimes \I \,} \leq 2 ( 1 + 3 \q ) \sqrt{\varepsilon}.
\end{equation*}
One last application of the reverse triangle inequality combined with the first result of Lemma~\ref{lem:A0-A1-commute} gives the final result.
\end{proof}
In the last two lemmas we bound the inner products appearing in the fidelity expression given in Eq.~\eqref{eq:fidelity-lower-bound}.
\lemak{lem:Cx-bound}{the inner product $\ave{ C_{\X}, \rho_{AB}}$ is close to unity}{
\begin{equation*}
\avrh{ C_{\X} }{} \geq 1 - \big( 7 + 5 \q \big) \sqrt{\varepsilon}.
\end{equation*}
}
\begin{proof}
Let
\begin{equation*}
K := 3 ( A_{0} + A_{1} ) - ( A_{1} A_{0} A_{1} + A_{0} A_{1} A_{0} ) - 3 A_{2} ( A_{0} + A_{1} ) A_{2} + A_{2} ( A_{1} A_{0} A_{1} + A_{0} A_{1} A_{0} ) A_{2}
\end{equation*}
and note that the operator $C_{\X}$ can be written as
\begin{equation*}
C_{\X} = \frac{1}{64} K \otimes \big[ 4 ( B_{0} + B_{1} ) - ( B_{0} + B_{1} B_{0} B_{1} ) - ( B_{1} + B_{0} B_{1} B_{0} ) \big].
\end{equation*}
Therefore,
\begin{equation*}
\ave{ C_{\X}, \rho_{AB} } = \frac{1}{16} \ave[\big]{ K \otimes ( B_{0} + B_{1} ), \rho_{AB} } - \frac{1}{64} \ave[\big]{ K \otimes \big[ ( B_{0} + B_{1} B_{0} B_{1} ) + ( B_{1} + B_{0} B_{1} B_{0} ) \big], \rho_{AB}}.
\end{equation*}
The second term we can already bound since the Cauchy--Schwarz inequality and the fact that $( B_{0} + B_{1} B_{0} B_{1} )^{2} = ( B_{1} + B_{0} B_{1} B_{0} )^{2} = \{ B_{0}, B_{1} \}^{2}$ imply that
\begin{equation*}
\abs[\big]{ \avrh{ K \otimes ( B_{0} + B_{1} B_{0} B_{1} ) }{\big} } \leq \fnrhh{ K \otimes \I \,} \cdot \fnrhh{ \I \otimes \{ B_{0}, B_{1} \} }.
\end{equation*}
The first factor can be bounded by
\begin{equation*}
\fnrhh{ K \otimes \I \,} = \sqrt{ \avrh{ K^{2} }{} } \leq \norm{ K }_{\infty} \leq 16.
\end{equation*}
In the second step we write $K = 8 K_{0} + K_{1} + K_{2} - 4 K_{3} - K_{4} - K_{5}$, where
\begin{align*}
K_{0} &:= A_{0} + A_{1},\\
K_{1} &:= A_{2} ( A_{0} +  A_{1} A_{0} A_{1} ) A_{2},\\
K_{2} &:= A_{2} ( A_{1} +  A_{0} A_{1} A_{0} ) A_{2},\\
K_{3} &:= ( A_{0} + A_{1} ) + A_{2} ( A_{0} + A_{1} ) A_{2},\\
K_{4} &:= A_{0} +  A_{1} A_{0} A_{1},\\
K_{5} &:= A_{1} + A_{0} A_{1} A_{0}.
\end{align*}
Note that $\avrh{ K_{0} \otimes ( B_{0} + B_{1} ) }{} = \avrh{ ( A_{0} + A_{1} ) \otimes ( B_{0} + B_{1} ) }{}$ is precisely the term we have bounded in Lemma~\ref{lem:(A0+A1)x(B0+B1)}. To show that all the other terms approximately vanish we apply inequalities~\eqref{eq:cauchy-schwarz-spec} and \eqref{eq:norm-inequality} to obtain
\begin{equation*}
\abs[\big]{ \avrh{ K_{j} \otimes ( B_{0} + B_{1} ) }{\big} } \leq \fnrhh{ K_{j} \otimes ( B_{0} + B_{1} ) \,} \leq 2 \fnrhh{ K_{j} \otimes \I \,}.
\end{equation*}
For $j = 1, 2$ we have
\begin{equation*}
\fnrhh{ K_{j} \otimes \I \,} = \sqrt{ \ave{ A_{2} \{ A_{0}, A_{1} \}^{2} A_{2}, \rho_{AB} } } = \fnrhh{ A_{2} \{A_{0}, A_{1} \} \otimes \I \,} \leq \fnrhh{ \{A_{0}, A_{1} \} \otimes \I \,},
\end{equation*}
For $j = 3$ we use inequality~\eqref{eq:flip-identity} to obtain
\begin{equation*}
\fnrhh{ K_{3} \otimes \I \,} = \fnrhh{ \{A_{0} + A_{1}, A_{2} \} \otimes \I \,}.
\end{equation*}
Similarly, for $j = 4, 5$ we have
\begin{equation*}
\fnrhh{ K_{j} \otimes \I \,} = \fnrhh{ \{A_{0}, A_{1} \} \otimes \I \,}.
\end{equation*}
Collecting all the error terms and plugging in the bounds derived in Lemmas~\ref{lem:A0-A1-commute} and~\ref{lem:A0+A1-A2-commute} and Corollary~\ref{cor:B0-B1-commute} leads to the desired inequality.
\end{proof}
\lemak{lem:Cz-bound}{the inner product $\avrh{ C_{\Z} }{}$ is close to unity}{
\begin{equation*}
\avrh{ C_{\Z} }{} \geq 1 - \frac{1}{2} \big( 4 + \q \big) \sqrt{\varepsilon}.
\end{equation*}
}
\begin{proof}
Note that the expression for $C_{\Z}$ can be written as
\begin{equation*}
C_{\Z} = \frac{1}{4 \q} \, A_{2} \otimes \big[ 4 ( B_{0} - B_{1} ) - ( B_{0} + B_{1} B_{0} B_{1} ) + ( B_{1} + B_{0} B_{1} B_{0} ) \big],
\end{equation*}
which immediately implies that
\begin{align*}
\avrh{ C_{\Z} }{} &= \frac{1}{\q} \avrh{ A_{2} \otimes ( B_{0} - B_{1} ) }{} - \frac{1}{4 \q } \avrh{ A_{2} \otimes ( B_{0} + B_{1} B_{0} B_{1} ) }{}  + \frac{1}{4 \q } \avrh{ A_{2} \otimes ( B_{1} + B_{0} B_{1} B_{0} ) }{}\\
&\geq \frac{1}{\q} \avrh{ A_{2} \otimes ( B_{0} - B_{1} ) }{} - \frac{1}{2 \q} \fnrhh{ \I \otimes \{ B_{0}, B_{1} \} \,},
\end{align*}
where we have used the Cauchy--Schwarz inequality combined with the observation that $( B_{0} + B_{1} B_{0} B_{1} )^{2} = ( B_{1} + B_{0} B_{1} B_{0} )^{2} = \{ B_{0}, B_{1} \}^{2}$. Plugging in the bounds derived in Corollary~\ref{cor:B0-B1-commute} gives the final result of the lemma.
\end{proof}
\subsection{Details of the numerical calculation using the swap method}
\label{app:swap-method}
We construct a $100 \times 100$ moment matrix $\Gamma$, whose rows and columns correspond to $P_{j} \ket{\psi}$, where $P_{j}$ is a monomial from the set $\{ \I, A_{x}, A_{x} A_{x'} \} \otimes \{ \I, B_{y}, B_{y} B_{y'} \}$. We impose the equality conditions resulting from $A_{x}^{2} = \I$ and $B_{y}^{2} = \I$, the normalisation condition $\Gamma_{jj} = 1$ for all $j$ and positivity $\Gamma \geq 0$. Then we minimise
\begin{equation}
\avrh{ - C_{\Y} + C_{\Z} }{}
\end{equation}
subject to a fixed Bell violation $\beta = t$ for various values of $t \in [5.7, 6]$. Inequality~\eqref{eq:fidelity-lower-bound2} leads to the lower bound on the fidelity presented in Fig.~1 in the main text.

Note that this moment matrix is not sufficient to obtain a bound on $\avrh{C_{\X}}{}$, because it does not contain strings of $A_{x}$ operators of sufficient length. Therefore, if we want to bound the fidelity using the inequality given in Eq.~\eqref{eq:fidelity-lower-bound} or bound all three terms simultaneously, we must construct a larger moment matrix. While we have been able to construct a larger moment matrix, we were not able to perform the numerical optimisation on it.
\section{Randomness certification}
\label{app:randomness}
In this appendix we explain the approach we have used to study the amount of randomness generated by the Bell inequality corresponding to $\alpha = 1$.
\subsection{The trade-off between marginals and the Bell violation}
We consider the simplest device-independent scenario: the devices of Alice and Bob are produced by Eve whose goal is to predict the outcome of Alice for a particular fixed setting. It is well-known that in this case Eve does not gain anything by entangling herself with the device. In fact, if we only care about her guessing probability, she does not even need to keep any classical knowledge about the device. Certifying randomness reduces to investigating the trade-off between the bias of the local observables and the observed Bell violation and such trade-offs can be studied numerically using the Navascu{\'e}s--Pironio--Ac{\'i}n (NPA) hierarchy~\cite{navascues07a, navascues08a}. More specifically, we use the ``1 + AB'' level to investigate the maximal bias of $A_{x}$ for a fixed violation $\beta$ (by symmetry the same constraints apply to the observables of Bob). We construct a $16 \times 16$ moment matrix $\Gamma$ whose rows and columns correspond to $\{ \ket{ \psi }, A_{x} \otimes \I \ket{ \psi }, \I \otimes B_{y} \ket{ \psi }, A_{x} \otimes B_{y} \ket{ \psi } \}$. We maximise the expectation value $\ave{ A_{x} }$ subject to a fixed Bell violation $\beta = t$ for various choices of $t \in [4, 5]$. To find feasible points we start with some optimal arrangement of the observables for Alice and Bob (as given in the main text) and consider a tilted version of the Bell operator: $r A_{x} \otimes \I + W$ for some $r \geq 0$. Finding the eigenvector corresponding to the largest eigenvalue gives a particular realisation for which typically $\beta = \ave{W} < 5$ and $\ave{A_{x}} > 0$. By generating a sufficient number of points and then taking their convex hull we construct the lower curves presented in Fig.~\ref{fig:randomness}.
Clearly, the upper and lower bounds turn out to be relatively close and they are consistent with our analytic result that the maximal violation certifies maximal randomness. Randomness produced by $A_{0}$ and $A_{1}$ can be certified all the way down to the classical value $\beta = 4$. Randomness of $A_{2}$, on the other hand, is only guaranteed for $\beta > 2 \sqrt{5}$ and we have indeed found a quantum realisation that achieves $\beta = 2 \sqrt{5}$ while keeping $A_{2}$ deterministic (see below). Clearly, $A_{0}$ and $A_{1}$ are better suited for generating randomness than $A_{2}$.
\begin{figure}[h]
\subfloat[]{\includegraphics{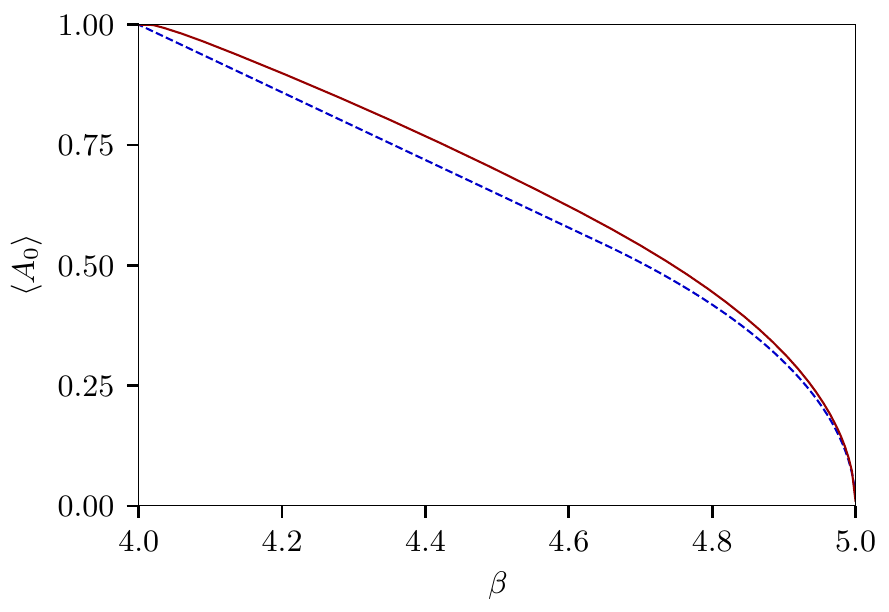}}
\subfloat[]{\includegraphics{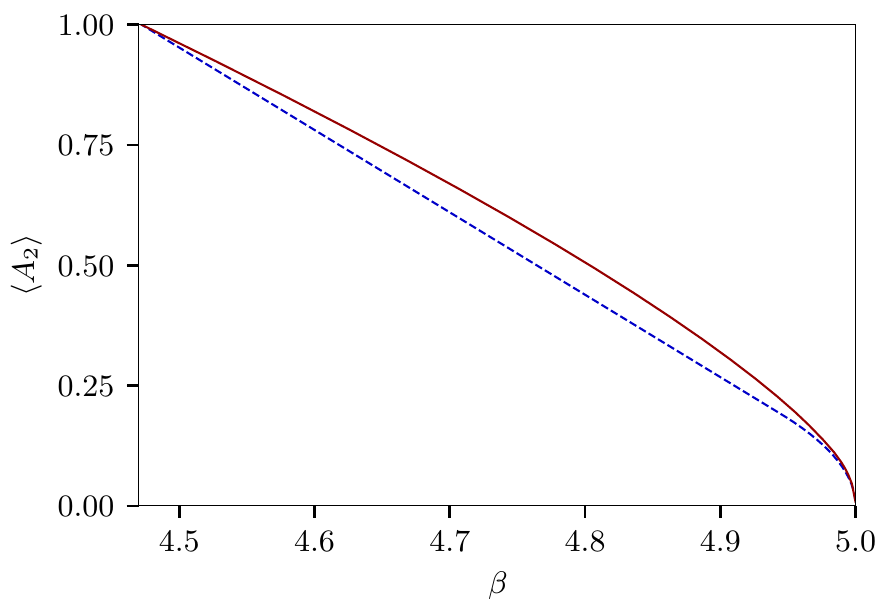}}
\caption{Numerical bounds on $\ave{A_{0}}$ and $\ave{A_{2}}$ for a fixed violation $\beta$ are shown in panels (a) and (b), respectively. The solid red lines represent the upper bounds obtained from the NPA hierarchy, while the dashed blue lines correspond to feasible points. The observable $A_{2}$ can only be used to generate randomness for violations exceeding $2 \sqrt{5}$.}
\label{fig:randomness}
\end{figure}
\subsection{Maximal violation under commutation constraints}
\label{app:maximal-violation-commutation}
Determining the minimal value of $\beta$ for which the observable $A_{x}$ is guaranteed to generate randomness is equivalent to finding the largest Bell value consistent with $\ave{ A_{x} } = \pm 1$. A related question is to determine the maximal value of $\beta$ under the assumption that certain observables commute, a problem that can be tackled numerically by imposing some additional constraints on the moment matrix. Numerical evidence suggests that:
\begin{align*}
[A_{0}, A_{1}] &= 0 \implies \beta \leq 2 \sqrt{5} \approx 4.47,\\
[A_{0}, A_{2}] &= 0 \implies \beta \leq \frac{ 2 + 3 \sqrt{6} }{2} \approx 4.67,\\
[A_{1}, A_{2}] &= 0 \implies \beta \leq \frac{ 2 + 3 \sqrt{6} }{2} \approx 4.67,\\
[A_{0}, A_{2}] = [A_{1}, A_{2}] &= 0 \implies \beta \leq \frac{ 2 + 3 \sqrt{6} }{2} \approx 4.67,\\
[A_{0}, A_{1}] = [A_{0}, A_{2}] &= 0 \implies \beta \lesssim 4.163,\\
[A_{0}, A_{1}] = [A_{1}, A_{2}] &= 0 \implies \beta \lesssim 4.163.
\end{align*}
Except for the last two cases we can provide explicit two-qubit realisations that saturate these bounds (see below). These results suggest that this Bell inequality can be used to make device-independent conclusions about the incompatibility structure of the employed observables~\cite{quintino19a}.

For $[A_{0}, A_{1}] = 0$ consider the observables
\begin{align*}
A_{0} &= \X, &\quad B_{0} &= \frac{ 2 \X + \Z }{\sqrt{5}}\\
A_{1} &= \X, &\quad B_{1} &= \frac{ 2 \X - \Z }{\sqrt{5}}\\
A_{2} &= \Z, &\quad B_{2} &= \I.
\end{align*}
It is easy to verify that $\ave{ W, \Phi^{+} } = 2 \sqrt{5}$. Note that this realisation also satisfies $[B_{0}, B_{2}] = [B_{1}, B_{2}] = 0$. Moreover, it shows that the value $\beta = 2 \sqrt{5}$ is consistent with $\ave{ B_{2} } = 1$ (and by symmetry with $\ave{ A_{2} } = 1$).

For $[A_{0}, A_{2}] = 0$ consider the observables
\begin{align*}
A_{0} &= \X, &\quad B_{0} &= \frac{ 9 \X + \sqrt{15} \, \Z }{4 \sqrt{6}}\\
A_{1} &= \frac{\X + \sqrt{15} \, \Z}{4}, &\quad B_{1} &= \frac{\X + \sqrt{15} \, \Z}{4}\\
A_{2} &= \X, &\quad B_{2} &= \frac{ \sqrt{3} \, \X - \sqrt{5} \, \Z }{ 2 \q }.
\end{align*}
It is easy to verify that $\ave{ W, \Phi^{+} } = ( 2 + 3 \sqrt{6} )/2$. A realisation satisfying $[A_{1}, A_{2}]$ can be obtained by swapping $A_{0} \leftrightarrow A_{1}$ and flipping the sign of $B_{2}$.
\bibliography{/home/jedrek/projekty/tex/library}
\end{document}